\documentclass[11pt]{article}
\usepackage{hyperref}
\usepackage{amssymb,amsmath,amsthm}
\usepackage{mathpazo}
\usepackage{algorithm}
\usepackage{algpseudocode}

\newtheorem{theorem}{Theorem}[section]

\newtheorem{definition}[theorem]{Definition}

\newtheorem{claim}[theorem]{Claim}
\newtheorem{lemma}[theorem]{Lemma}

\newtheorem{corollary}[theorem]{Corollary}
\newtheorem{fact}[theorem]{Fact}

\newtheorem{remark}[theorem]{Remark}

\def\Output{{Output}}
\def\calM{{\cal M}}
\def\calF{{\cal F}}
\def\calB{{\cal B}}
\def\calC{{\cal C}}

\def\L{{\cal L}}

\def\Z{{\mathbb{Z}}}
\def\R{\mathbb{R}}

\def\N{\mathbb{N}}

\newcommand\Prob[2]{{\Pr_{#1}\left[ {#2} \right]}}

\newcommand{\eps}{\epsilon}
\renewcommand{\epsilon}{\varepsilon}

\newcommand{\rank}{\mathop{\mathrm{rank}}}

\newcommand{\linspan}{\mathop{\mathrm{span}}}

\newcommand{\Fset}{\mathbb{F}}         

\newcommand{\SVP}{\mathsf{SVP}}

\newcommand{\LIP}{\mathsf{LIP}}
\newcommand{\GIP}{\mathsf{GIP}}

\newcommand{\NP}{\mathsf{NP}}

\newcommand{\coNP}{\mathsf{coNP}}

\newcommand{\AM}{\mathsf{AM}}
\newcommand{\coAM}{\mathsf{coAM}}
\newcommand{\SZK}{\mathsf{SZK}}
\newcommand{\HVSZK}{\mathsf{HVSZK}}

\newcommand{\YES}{\mathsf{YES}}
\newcommand{\NO}{\mathsf{NO}}

\begin{document}

\title{{\bf On the Lattice Isomorphism Problem}}

\author{
 Ishay Haviv\thanks{School of Computer Science, The Academic College of Tel Aviv-Yaffo, Tel Aviv 61083, Israel.
}
 \and
 Oded Regev\thanks{Courant Institute of Mathematical Sciences, New York University. This material is based upon work supported by the National Science Foundation under Grant No.~CCF-1320188. Any opinions, findings, and conclusions or recommendations expressed in this material are those of the authors and do not necessarily reflect the views of the National Science Foundation.
 }}

\date{}

\maketitle

\begin{abstract}

We study the {\em Lattice Isomorphism Problem} ($\LIP$), in which given two lattices $\L_1$ and $\L_2$ the goal is to decide whether there exists an orthogonal linear transformation mapping $\L_1$ to $\L_2$. Our main result is an algorithm for this problem running in time $n^{O(n)}$ times a polynomial in the input size, where $n$ is the rank of the input lattices. A crucial component is a new generalized {\em isolation lemma}, which can isolate $n$ linearly independent vectors in a given subset of $\Z^n$ and might be useful elsewhere. We also prove that $\LIP$ lies in the complexity class $\SZK$.
\end{abstract}


\section{Introduction}

An $m$-dimensional {\em lattice} $\L$ of rank $n$ is defined as the set of all integer combinations of $n$ linearly independent vectors $b_1,\ldots,b_n \in \R^m$, which form a {\em basis} of the lattice. This mathematical object, despite its simplicity, hides a rich geometrical structure, which was extensively studied in the last decades by the theoretical computer science community. This was initiated by the discovery of the famous LLL algorithm in 1982~\cite{LLL82} and was further motivated by Ajtai's cryptographic application of lattices in 1996~\cite{ajtai96generating}. To date, lattices have numerous applications in several areas of computer science including algorithms, computational complexity and cryptography.

One of the most fundamental lattice problems is the Shortest Vector Problem ($\SVP$), where given a lattice basis the goal is to find a shortest nonzero vector in the lattice. This problem is known to be $\NP$-hard (under randomized reductions) for approximation factors which are almost polynomial in the lattice rank $n$~\cite{Khot05svp,HRsvp,MicSVP12} and to be solved in its exact version by algorithms of running time exponential in $n$~\cite{Kannan87,MicV10}. However, $\SVP$ with approximation factors of $\sqrt{n/\log n}$ and $\sqrt{n}$ is known to be in $\coAM$ and in $\coNP$ respectively~\cite{Goldreich:1998:LNA,AharonovR04}, hence is not $\NP$-hard for these factors unless the polynomial time hierarchy collapses. A major challenge in the area is to understand how hard $\SVP$ and related lattice problems are for polynomial approximation factors, as this is what lattice-based cryptography relies on.

This paper is concerned with the Lattice Isomorphism Problem ($\LIP$). Two lattices $\L_1$ and $\L_2$ are {\em isomorphic} if there exists an orthogonal linear transformation mapping $\L_1$ to $\L_2$. In $\LIP$ one wishes to decide whether two given lattices are isomorphic or not. The problem was studied by Plesken and Souvignier~\cite{PleskenS97} (using ideas from the earlier work~\cite{PleskenP85}) who suggested algorithms that can solve the problem in low dimensions for specific lattices of interest. The asymptotic complexity of the problem was later considered by Dutour Sikiri\'{c}, Sch{\"u}rmann, and Vallentin~\cite{LatticeIso}, and it also showed up in cryptographic applications of lattices~\cite{Szydlo03}. Recently, Lenstra, Schoof, and Silverberg presented an efficient algorithm that can decide if a given lattice is isomorphic to $\Z^n$, assuming some information about its symmetries is provided as a hint~\cite{LenstraSS13}.

Deciding whether two given combinatorial or algebraic structures are isomorphic is a notorious question in the theory of computing. A well-known special case of this problem is the Graph Isomorphism Problem ($\GIP$), in which given two graphs $G_1$ and $G_2$ one has to decide whether there exists an edge-preserving bijection from the vertex set of $G_1$ to that of $G_2$. The best known worst-case running time of an algorithm for $\GIP$ is $2^{\widetilde{O}(\sqrt{n})}$, where $n$ stands for the number of vertices~\cite{BabaiIso}. It was shown in~\cite{GoldreichMW91} that $\GIP$ lies in the complexity class $\coAM$. This implies that, unless the polynomial time hierarchy collapses, $\GIP$ is not $\NP$-hard, and it is a long-standing open question whether there exists a polynomial time algorithm solving it (see, e.g.,~\cite{BabaiSurvey}). Interestingly, it was shown in~\cite{LatticeIso} that the isomorphism problem on lattices is at least as hard as that on graphs.

Another isomorphism problem of interest is the Code Equivalence Problem, in which given two $n$-dimensional linear codes $\calC_1$ and $\calC_2$ over some field $\Fset$ the goal is to decide whether there exists a permutation on the coordinates mapping $\calC_1$ to $\calC_2$. This problem was studied by Petrank and Roth~\cite{PetrankR97}, who showed that it lies in $\coAM$ and is at least as hard as $\GIP$. Recently, Babai showed an algorithm solving it in time $(2+o(1))^n$ (see~\cite[Appendix~7.1]{BabaiCGQ11}).

\subsection{Our Results}

Our main result is an algorithm that given two lattices computes {\em all} orthogonal linear transformations mapping one lattice to another and, in particular, decides $\LIP$.

\begin{theorem}\label{thm:MainAlgIntro}
There exists an algorithm that given two bases of lattices $\L_1$ and $\L_2$ of rank $n$, outputs all orthogonal linear transformations $O: \linspan(\L_1) \rightarrow \linspan(\L_2)$ for which $\L_2 = O(\L_1)$ in running time $n^{O(n)} \cdot s^{O(1)}$ and in polynomial space, where $s$ denotes the input size. In addition, the number of these transformations is at most $n^{O(n)}$.
\end{theorem}

\noindent
We note that the bound in Theorem~\ref{thm:MainAlgIntro} on the number of orthogonal linear transformations mapping one lattice to another is tight up to the constant in the exponent. To see this, observe that the isomorphisms from the lattice $\Z^n$ to itself are precisely all the $2^n \cdot n! = n^{\Omega(n)}$ sign permutations. This implies that the running time of the algorithm from Theorem~\ref{thm:MainAlgIntro} is optimal, up to the constant in the exponent, given that it outputs all isomorphisms between the two input lattices. However, the challenge of finding a more efficient algorithm which only decides $\LIP$ is left open.

The algorithm from Theorem~\ref{thm:MainAlgIntro} is crucially based on a new version of the celebrated {\em isolation lemma} of Valiant and Vazirani~\cite{ValiantV86}. A standard version of the lemma says that for every set $C \subseteq \Z^n$ of short vectors (in $\ell_\infty$ norm), most short integer vectors $z$ have a {\em single} vector in $C$ that minimizes the inner product with $z$ over all vectors in $C$. The isolation lemma has appeared in the literature in several variations for various applications, ranging from the design of randomized algorithms, e.g.,~\cite{MulmuleyVV87,NarayananSV94,ChariRS95,KlivansS01}, to results in computational complexity, e.g.,~\cite{Toda91,Wigderson94,ReinhardtA00,ArvindM08} (for a survey see~\cite{IsolationSurvey}). Whereas the lemma is usually used to isolate one vector, for our application we need to isolate $n$ linearly independent vectors in $C$. The lemma below guarantees the existence of a vector $z$ and a sequence of $n$ linearly independent vectors in $C$, each of which uniquely minimizes the inner product with $z$ over all vectors in $C$ which are not in the linear span of the previous ones.

\begin{lemma}\label{lemma:isolation_spanIntro}
Let $C \subseteq \Z^n$ be a set of vectors satisfying $\|c\|_\infty \leq K$ for every $c \in C$ and $\linspan(C) = \R^n$. Let $z=(z_1,\ldots,z_n)$ be a random vector such that each $z_i$ is independently chosen from the uniform distribution over $\{1,\ldots,R\}$ for $R = K(2K+1)n^3/\eps$. Then, with probability at least $1-\eps$, there are $n$ linearly independent vectors $x_1,\ldots,x_n \in C$ such that for every $1 \leq j \leq n$, the minimum inner product of $z$ with vectors in $C \setminus \linspan(x_1,\ldots,x_{j-1})$ is uniquely achieved by $x_j$.
\end{lemma}

\noindent
We actually prove this in a more general setting, in which $\linspan$ can be replaced by any function satisfying some condition. This more general statement includes as special cases some of the previously known variations of the isolation lemma, and might be useful elsewhere. See Section~\ref{sec:isolation} for details.

Finally, we prove that $\LIP$, which naturally lies in $\NP$, has a statistical zero-knowledge proof system and hence belongs to the complexity class $\SZK$. This result was independently observed by Greg Kuperberg~\cite{Kup}.

\begin{theorem}\label{thm:coAMIntro}
$\LIP$ is in $\SZK$.
\end{theorem}

\noindent
It is well known that $\SZK \subseteq \AM \cap \coAM$~\cite{Fortnow89,AielloH91}. As a result, just like many other lattice problems (e.g., the problem of approximating the length of a shortest nonzero vector to within polynomial factors, which is central in lattice-based cryptography), $\LIP$ is unlikely to be $\NP$-hard. We note, though, that the reduction from the Graph Isomorphism Problem ($\GIP$)~\cite{LatticeIso} gives some evidence that $\LIP$ is a hard problem, evidence that is lacking for other lattice problems.

\subsection{Overview of Proofs and Techniques}

\subsubsection{The Algorithm for \texorpdfstring{$\LIP$}{LIP}}

The input of $\LIP$ consists of two lattices $\L_1$ and $\L_2$ of rank $n$, and the goal is to decide if there exists an orthogonal linear transformation $O$ satisfying $\L_2=O(\L_1)$. In order to find such an $O$ it suffices to find $n$ linearly independent vectors in $\L_1$ and their image in $\L_2$ according to $O$. Since $O$ preserves lengths, a possible approach is to compute $n$ linearly independent short vectors of $\L_1$ and try to map them to all $n$-tuples of short vectors of $\L_2$.

Consider the case where the lattices $\L_1$ and $\L_2$ have only one shortest nonzero vector (up to sign). In this case, there are only two possible choices for how an isomorphism from $\L_1$ to $\L_2$ can act on these vectors. Hence, one can recursively solve the problem by considering the lattices $\L_1$ and $\L_2$ projected to the spaces orthogonal to their shortest vectors. This demonstrates that the hard instances of the problem are those where the lattices have $n$ linearly independent shortest vectors, so in the rest of this discussion let us assume that we are in this case.

Given the lattices $\L_1$ and $\L_2$ it is possible to compute the sets $A_1$ and $A_2$ of all shortest nonzero vectors in $\L_1$ and $\L_2$ respectively. Indeed, by the algorithm of~\cite{MicV10}, the running time needed to compute $A_1$ and $A_2$ is $2^{O(n)}$ (or $n^{O(n)}$, if we insist on polynomial space complexity~\cite{Kannan87}). Now, consider the algorithm that for certain $n$ linearly independent vectors in $A_1$ tries all the linear transformations that map them to $n$ linearly independent vectors in $A_2$ and checks if at least one of them is orthogonal and maps $\L_1$ to $\L_2$. Notice that the running time of this algorithm crucially depends on the number of shortest nonzero vectors in the lattices, which is usually referred to as their {\em kissing number}. It is easy to see that the kissing number of a lattice of rank $n$ is at most $2^{n+1}$.\footnote{Indeed, if there are more than $2^{n+1}$ shortest nonzero lattice vectors then at least two of them belong to the same coset of $2\L$ and have nonzero average, hence their average is a shorter lattice vector, in contradiction.} This implies that the suggested algorithm has running time whose dependence on $n$ is bounded by $2^{O(n^2)}$. However, the true worst-case running time of this algorithm is a function of the maximum possible kissing number of a lattice of rank $n$, whose value is an open question. The best currently known lower bound is $n^{\Omega(\log n)}$~\cite{BS83} (see also~\cite[Page~151]{ConwaySloane}), hence even if this lower bound were tight (which does not seem particularly likely), the algorithm would run in time $n^{\Omega(n\log n)}$, which is still asymptotically slower than our algorithm.

We improve on the above naive algorithm by showing a way to isolate $n$ linearly independent vectors in the sets $A_1$ and $A_2$. In Theorem~\ref{thm:uniquelyL} we prove, using our isolation lemma (Lemma~\ref{lemma:isolation_spanIntro}), that for a lattice $\L_1$ as above there exists a relatively short vector $v$ in the dual lattice $\L_1^*$ that {\em uniquely} defines $n$ linearly independent vectors $x_1,\ldots,x_n$ in $A_1$. These vectors are defined as follows: for every $1 \leq j \leq n$, the minimum inner product of $v$ with vectors in $A_1 \setminus \linspan(x_1,\ldots,x_{j-1})$ is uniquely achieved by $x_j$. Given such a $v$ (which can be found by enumerating all short vectors in $\L_1^*$), we try all vectors of norm $\|v\|$ in $\L_2^*$, of which there are at most $n^{O(n)}$.
Once we find the image of $v$ under $O$, we use it to apply the same process as above with $A_2$ obtaining $n$ linearly independent vectors in $A_2$. Since $O$ preserves inner products, these vectors must be the images of $x_1,\ldots,x_n$ under $O$, which allows us to find $O$.

\subsubsection{\texorpdfstring{$\LIP$}{LIP} is in \texorpdfstring{$\SZK$}{SZK}}

We turn to discuss the proof of Theorem~\ref{thm:coAMIntro} which says that $\LIP$ lies in the complexity class $\SZK$. Since $\SZK$ is known to be closed under complement~\cite{Okamoto00}, it suffices to show a statistical zero-knowledge proof system that enables an efficient verifier to verify that two given lattices are {\em not} isomorphic. The high level idea is similar to known proof systems of the complement of other isomorphism problems, e.g., Graph Isomorphism and Code Equivalence~\cite{ValiantV86,PetrankR97}. In these proof systems the verifier picks uniformly at random one of the two objects given as input and sends to the prover a random representative of its isomorphism class. The verifier accepts if and only if the prover identifies which of the two objects was chosen. A crucial observation is that if the two objects are isomorphic then the prover gets a sample which is independent of the chosen object, hence she is not able to identify it with probability higher than $1/2$. On the other hand, if the objects are not isomorphic, then the object sent by the verifier is isomorphic to exactly one of the two, hence a computationally unbounded prover is able to answer correctly. Moreover, the correct answer is known in advance to the verifier, who therefore does not learn anything new from the prover's answer.

In the lattice analogue of the above proof system, in addition to choosing a lattice that forms a representative of the isomorphism class, the verifier also has to choose a basis that generates the lattice. Observe that the basis should be chosen in a way that does not provide any useful information for the prover, and in particular, must not depend on the input bases. To deal with this difficulty we use known efficient algorithms to sample lattice vectors from the discrete Gaussian distribution~\cite{GentryPV08} (see also~\cite{BrakerskiLPRS13}), and prove that polynomially many samples suffice to obtain a generating set for the lattice (Lemma~\ref{lemma:SampleGenSet}). We can then send a random rotation of this set of vectors. In fact, to avoid issues of accuracy, we instead send the matrix of all pairwise inner products of these vectors (the Gram matrix).

\subsection{Outline}
The rest of the paper is organized as follows. In Section~\ref{sec:pre} we gather basic definitions and results we shall later use. In Section~\ref{sec:isolation} we prove our new generalized isolation lemma and derive Lemma~\ref{lemma:isolation_spanIntro}. In Section~\ref{sec:algorithm} we present our algorithm for $\LIP$ and prove Theorem~\ref{thm:MainAlgIntro}. This is done in two steps, where in the first we assume that the input lattices contain $n$ linearly independent shortest vectors (i.e., $\lambda_1=\lambda_n$), and in the second we extend the algorithm to the general case. Finally, in Section~\ref{sec:coAM} we prove Theorem~\ref{thm:coAMIntro}.

\section{Preliminaries}\label{sec:pre}

\subsection{General}
An {\em orthogonal} linear transformation (or {\em isometry}) $O : V_1 \rightarrow V_2$ is a linear transformation that preserves inner products, that is, $\langle x, y\rangle = \langle O(x), O(y)\rangle$ for every $x,y \in V_1$. For a set $A \subseteq V_1$ we use the notation $O(A) = \{ O(x) \mid x \in A\}$. For a matrix $B$ we denote its $i$th column by $b_i$, and $O(B)$ stands for the matrix whose $i$th column is $O(b_i)$. The {\em Gram matrix} of a matrix $B$ is defined to be the matrix $G = B^T \cdot B$, or equivalently, $G_{i,j} = \langle b_i , b_j \rangle$ for every $i$ and $j$. The Gram matrix of a matrix $B$ determines its columns up to an orthogonal linear transformation, as stated below. Note that $\linspan(B)$ stands for the subspace spanned by the columns of $B$.

\begin{fact}\label{fact:Gram}
Let $B$ and $D$ be two matrices satisfying $B^T \cdot B = D^T \cdot D$. Then there exists an orthogonal linear transformation $O:\linspan(B) \rightarrow \linspan(D)$ for which $D = O(B)$.
\end{fact}


\subsection{Lattices}

An $m$-dimensional {\em lattice} $\L \subseteq \R^m$ is the set of all integer combinations of a set of linearly independent vectors $\{b_1,\ldots,b_n\} \subseteq \R^m $, i.e., $\L=\{\sum_{i=1}^{n}{a_i b_i}~|~\forall i.~a_i \in \Z\}$. The set $\{b_1,\ldots,b_n\}$ is called a {\em basis} of $\L$ and $n$, the number of vectors in it, is the {\em rank} of $\L$. Let $B$ be the $m$ by $n$ matrix whose $i$th column is $b_i$. We identify the matrix and the basis that it represents and denote by $\L(B)$ the lattice that $B$ generates. The norm of a basis $B$ is defined by $\|B\| = \max_{i}{\|b_i\|}$.
A basis of a lattice is not unique. It is well known that two bases $B_1$ and $B_2$ generate the same lattice of rank $n$ if and only if $B_1 = B_2 \cdot U$ for a {\em unimodular} matrix $U \in \Z^{n \times n}$, i.e., an integer matrix satisfying $|\det(U)|=1$. The determinant of a lattice $\L$ is defined by $\det(\L)=\sqrt{\det(B^T B)}$, where $B$ is a basis that generates $\L$. It is not difficult to verify that $\det(\L)$ is independent of the choice of the basis. A set of (not necessarily linearly independent) vectors that generate a lattice is called a {\em generating set} of the lattice. A lattice $\calM$ is a {\em sublattice} of a lattice $\L$ if $\calM \subseteq \L$, and it is a {\em strict sublattice} if $\calM \subsetneq \L$. If a lattice $\L$ and its sublattice $\calM$ span the same subspace, then the {\em index} of $\calM$ in $\L$ is defined by $|\L : \calM| = \det(\calM)/\det(\L)$. It is easy to see that if $\calM$ is a sublattice of $\L$ such that $|\L : \calM|=1$ then $\calM=\L$.

The length of a shortest nonzero vector in $\L$ is denoted by $\lambda_1(\L) = \min\{\|u\|~|~u \in \L \setminus \{0\}\}$. The following simple and standard fact provides an upper bound on the number of short vectors in a lattice of rank $n$ (see, e.g.,~\cite{MicV10}).

\begin{fact}\label{fact:short}
For every lattice $\L$ of rank $n$ and for every $t \geq 0$, the number of vectors in $\L$ of norm at most $t \cdot \lambda_1(\L)$ is at most $(2t+1)^n$.
\end{fact}

\begin{proof}
Consider all the (open) balls of radius $\lambda_1(\L)/2$ centered at the lattice points of distance at most $t \cdot \lambda_1(\L)$ from the origin. These balls are pairwise disjoint and are all contained in the ball centered at the origin whose radius is $(t+1/2) \cdot \lambda_1(\L)$. This implies that their number is at most \[\Big(\frac{(t+1/2) \cdot \lambda_1(\L)}{\lambda_1(\L)/2}\Big)^n = (2t+1)^n. \qedhere \]
\end{proof}

The definition of $\lambda_1$ is naturally extended to the {\em successive minima} $\lambda_1,\ldots,\lambda_n$ defined as follows:
\[\lambda_i(\L) = \inf \{r>0~|~\rank(\linspan(\L \cap (r \cdot\calB))) \geq i\},\]
where $\calB$ denotes the ball of radius $1$ centered at the origin.
A somewhat related lattice parameter, denoted $bl(\L)$, is defined as the minimum norm of a basis that generates $\L$. It is known that $bl$ is related to the $n$th successive minimum by $\lambda_n(\L) \leq bl(\L) \leq \frac{\sqrt{n}}{2} \cdot \lambda_n(\L)$ (see, e.g.,~\cite{CaiN97}).

As mentioned before, two lattices $\L_1$ and $\L_2$ are {\em isomorphic} if there exists an orthogonal linear transformation $O:\linspan(\L_1) \rightarrow \linspan(\L_2)$ for which $\L_2=O(\L_1)$. In this paper we study the computational problem, called the {\em Lattice Isomorphism Problem} ($\LIP$), defined as follows. The input consists of two lattices $\L_1$ and $\L_2$ and we are asked to decide if they are isomorphic or not. One subtle issue is how to specify the input to the problem. One obvious way is to follow what is commonly done with other lattice problems, namely, the lattices are given as a set of basis vectors whose entries are given as rational numbers. This however leads to what we feel is an unnecessarily restricted definition: orthogonal matrices typically involve irrational entries, hence bases of two isomorphic lattices will typically also include irrational entries. Such bases, however, cannot be specified exactly as an input. Instead, we follow a much more natural definition (which is clearly as hard as the previous one, making our results stronger) in which the input bases are specified in terms of their \emph{Gram matrices}. Notice that a Gram matrix specifies a basis only up to rotation, but this is clearly inconsequential for $\LIP$.

\begin{definition}
In the Lattice Isomorphism Problem ($\LIP$) the input consists of two Gram matrices $G_1$ and $G_2$, and the goal is to decide if there exists a unimodular matrix $U$ for which $G_1 = U^T \cdot G_2 \cdot U$.
\end{definition}

For clarity, in our algorithms we assume that the input is given as a basis, and we ignore issues of precision. This is justified because (1) an arbitrarily good approximation of a basis can be extracted from a Gram matrix using the Cholesky decomposition; and (2) given a good enough approximation of a purported orthogonal transformation it is possible to check if it corresponds to a true lattice isomorphism by extracting the corresponding (integer-valued) unimodular matrix $U$ that converts between the bases and checking the equality $G_1=U^T G_2 U$, which only involves exact arithmetic. We note that an alternative, possibly more disciplined, solution is to avoid working with lattice vectors directly and instead work with their integer coefficients in terms of a lattice basis, and use the Gram matrix to compute norms and inner products (see, e.g.,~\cite[Page~80]{Cohen93}).

\subsection{Dual Lattices}
The {\em dual lattice} of a lattice $\L$, denoted by $\L^*$, is defined as the set of all vectors in $\linspan(\L)$ that have integer inner
product with all the lattice vectors of $\L$, that is,
\[\L^* = \{u \in \linspan(\L)~|~\forall v \in \L.~\langle u,v \rangle \in \Z \}.\]
The {\em dual basis} of a lattice basis $B$ is denoted by $B^*$ and is defined as the one which satisfies $B^T \cdot B^* = I$ and $\linspan(B)=\linspan(B^*)$, that is, $B^* = B(B^T B)^{-1}$. It is well known that the dual basis generates the dual lattice, i.e., $\L(B)^* = \L(B^*)$.

In~\cite{Bana93} Banaszczyk proved relations between parameters of lattices and parameters of their dual. Such results are known as {\em transference theorems}. One of his results, which is known to be tight up to a multiplicative constant, is the following.
\begin{theorem}[\cite{Bana93}]\label{thm:transference}
For every lattice $\L$ of rank $n$, $1 \leq
\lambda_1(\L) \cdot \lambda_n(\L^*) \leq n$.
\end{theorem}

\subsection{Korkine-Zolotarev Bases}
Before defining Korkine-Zolotarev bases we need to define the {\em Gram-Schmidt orthogonalization process}. For a sequence of vectors $b_1,\ldots,b_n$ define the corresponding Gram-Schmidt orthogonalized vectors $\tilde{b}_1,\ldots,\tilde{b}_n$ by
\begin{eqnarray*}
\tilde{b}_i = b_i -\sum_{j=1}^{i-1}{\mu_{i,j}\tilde{b}_j},~~~\mu_{i,j} = \frac{\langle b_i, \tilde{b}_j \rangle}{\langle \tilde{b}_j, \tilde{b}_j \rangle}.
\end{eqnarray*}
In words, $\tilde{b}_i$ is the component of $b_i$ orthogonal to $b_1,\ldots,b_{i-1}$. A Korkine-Zolotarev basis is defined as follows.
\begin{definition}\label{def:KZ}
Let $B$ be a basis of a lattice $\L$ of rank $n$ and let $\widetilde{B}$ be the corresponding Gram-Schmidt orthogonalized basis. For $1 \leq i \leq n$ define the projection function $\pi^{(B)}_i(x)=\sum_{j = i}^{n}{\langle x,\tilde{b}_j \rangle \cdot \tilde{b}_j / \|\tilde{b}_j\|^2}$ that maps $x$ to its projection on
$\linspan(\tilde{b}_i,\ldots,\tilde{b}_n)$. A basis $B$ is a {\em Korkine-Zolotarev basis} if for all $1 \leq i \leq n$,
\begin{itemize}
    \item $\tilde{b}_i$ is a shortest nonzero vector in $\pi^{(B)}_i(\L)=\{\pi_i^{(B)}(u)~|~u \in \L\}$,
    \item and for all $j<i$, the Gram-Schmidt coefficients $\mu_{i,j}$
    of $B$ satisfy $|\mu_{i,j}| \leq \frac{1}{2}$.
\end{itemize}
A basis $B$ is a {\em dual Korkine-Zolotarev basis} if its dual $B^*$ is a Korkine-Zolotarev basis.
\end{definition}

Lagarias, Lenstra and Schnorr~\cite{LLS90} related the norms of the vectors in a Korkine-Zolotarev basis to the successive minima of the lattice, as stated below.
\begin{theorem}[\cite{LLS90}]\label{thm:KZ}
If $B$ is a Korkine-Zolotarev basis of a lattice $\L$ of rank $n$, then for all $1 \leq i \leq n$,
\[\|b_i\| \leq \sqrt{i} \cdot \lambda_i(\L).\]
\end{theorem}

The following lemma provides an upper bound on the coefficients of short lattice vectors in terms of a dual Korkine-Zolotarev basis.

\begin{lemma}\label{lemma:coef}
Let $B$ be a dual Korkine-Zolotarev basis of a lattice $\L$ of rank $n$ and let $a_1,\ldots,a_n$ be integer coefficients of a vector $v = \sum_{i=1}^{n}{a_i \cdot b_i}$ of $\L$ satisfying $\|v\| \leq t \cdot \lambda_1(\L)$. Then, for every $1 \leq i \leq n$, $|a_i| \leq t \cdot n^{3/2}$.
\end{lemma}

\begin{proof}
Since the dual basis $B^*$ satisfies $B^T \cdot B^* = I$, it follows that $a_i = \langle v,b^*_i \rangle$ for every $1 \leq i \leq n$. By the Cauchy-Schwarz inequality, we obtain that
\[|a_i| = |\langle v,b^*_i \rangle| \leq \|v\| \cdot \|b^*_i\| \leq t \cdot \lambda_1(\L) \cdot \sqrt{n} \cdot \lambda_n(\L^*) \leq t \cdot n^{3/2},\]
where the second inequality follows from Theorem~\ref{thm:KZ}, and the third one from Theorem~\ref{thm:transference}.
\end{proof}

\begin{lemma}\label{lemma:normK}
Let $B$ be a dual Korkine-Zolotarev basis of a lattice $\L$ of rank $n$ satisfying $\lambda_1(\L)=\lambda_n(\L)$, and let $B^*$ be its dual. Then, for every integer coefficients $a_1,\ldots,a_n$ satisfying $|a_i| \leq K$ for every $1 \leq i \leq n$, the vector $v = \sum_{i=1}^{n}{a_i \cdot b^*_i} \in \L^*$ satisfies $\|v\| \leq n^{5/2} \cdot K \cdot \lambda_1(\L^*)$.
\end{lemma}

\begin{proof}
First, use Theorem~\ref{thm:transference} twice to obtain \[\lambda_n(\L^*) \leq \frac{n}{\lambda_1(\L)} = \frac{n}{\lambda_n(\L)} \leq n \cdot \lambda_1(\L^*).\]
Now, by the triangle inequality and Theorem~\ref{thm:KZ} applied to the Korkine-Zolotarev basis $B^*$, it follows that
\[\|v\| \leq \sum_{i=1}^{n}{|a_i| \cdot \|b^*_i\|} \leq n \cdot K \cdot \sqrt{n} \cdot \lambda_n(\L^*) \leq n^{5/2} \cdot K \cdot \lambda_1(\L^*). \qedhere \]
\end{proof}

\subsection{Gaussian Measures on Lattices}

For $n \in \N$ and $s>0$ let $\rho_s:\R^n \rightarrow (0,1]$ be the {\em Gaussian function} centered at the origin scaled by a factor of $s$ defined by
\[\forall x \in \R^n.~~\rho_{s}(x)=e^{-\pi \|x/s\|^2}.\]
We define the {\em discrete Gaussian distribution} with parameter $s$ on a lattice $\L$ of rank $n$ by its probability function
\[\forall x \in \L.~~D_{\L,s}(x)=\frac{\rho_{s}(x)}{\rho_{s}(\L)},\]
where for a set $A$ we denote $\rho_s(A) = \sum_{x \in A}{\rho_s(x)}.$
Notice that the sum $\rho_{s}(\L)$ over all lattice vectors is finite, as follows from the fact that $\int_{\R^n}{\rho_s(x){\rm d}x}=s^n$. It can be shown that $D_{\L,s}$ has expectation zero and expected squared norm close to $s^2 n/2\pi$ if $s$ is large enough. We need the following concentration result of Banaszczyk~\cite{Bana93}.

\begin{lemma}[\cite{Bana93},~Lemma~1.5(i)]\label{lemma:BanaNorm}
Let $\L$ be a lattice of rank $n$, and let $u$ be a vector chosen from $D_{\L,s}$. Then, the probability that $\|u\| \geq s \cdot \sqrt{n}$ is $2^{-\Omega(n)}$.
\end{lemma}

We also need the following simple claim, which follows from techniques in~\cite{Bana93}.

\begin{claim}[\cite{Bana93}]\label{claim:GaussianBound}
For every $n$-dimensional lattice $\L$, a real $s >0$ and a vector $w \in \R^n$, \[\rho_s(w+\L) \geq \rho_s(w) \cdot \rho_s(\L).\]
\end{claim}

\begin{proof}
The claim follows from the following calculation.
\begin{eqnarray*}
\rho_s(w+\L) &=& \sum_{x \in \L}{e^{-\pi \|w+x\|^2/s^2}} = \frac{1}{2} \cdot \sum_{x \in \L}{\Big(e^{-\pi \|x+w\|^2/s^2} + e^{-\pi \|x-w\|^2/s^2}\Big)} \\ &=& e^{-\pi \|w\|^2/s^2} \cdot \sum_{x \in \L}{\Big(e^{-\pi \|x\|^2/s^2} \cdot \cosh(2\pi \langle x,w \rangle/s^2) \Big)} \geq \rho_s(w) \cdot \rho_s(\L),
\end{eqnarray*}
where the inequality holds since  $\cosh(\alpha) \geq 1$ for every $\alpha$.
\end{proof}

The problem of efficient sampling from the discrete Gaussian distribution was studied by Gentry, Peikert and Vaikuntanathan~\cite{GentryPV08}. They showed a sampling algorithm whose output distribution is statistically close to the discrete Gaussian distribution on a given lattice, assuming that the parameter $s$ is sufficiently large. For convenience, we state below a recent result of Brakerski et al.~\cite{BrakerskiLPRS13} providing an {\em exact} sampling algorithm from the discrete Gaussian distribution.

\begin{lemma}[\cite{BrakerskiLPRS13},~Lemma~2.3]\label{lemma:SampleD}
There exists a probabilistic polynomial time algorithm $\mathsf{SampleD}$ that given a basis $B$ of a lattice $\L$ of rank $n$ and $s \geq \max_{i}{\|\tilde{b}_i\|}\cdot \sqrt{\ln (2n+4)/\pi}$ outputs a sample distributed according to $D_{\L,s}$.
\end{lemma}

\subsection{Lattice Algorithms}

The following two lemmas provide efficient algorithms for computing lattice bases. In the first, the lattice is given by a generating set, and in the second it is given as an intersection of a lattice and a subspace. Both algorithms are based on what is known as matrices of Hermite normal form (see, e.g.,~\cite[Chapter~8]{MicciancioBook}).

\begin{lemma}\label{lemma:MinimalLattice}
There is a polynomial time algorithm that given a set of vectors computes a basis for the lattice that they generate.
\end{lemma}

\begin{lemma}[\cite{MicSODA08},~Lemma~1]\label{lemma:MicSubspace}
There is a polynomial time algorithm that given a basis of an $m$-dimensional lattice $\L$ and a subspace $S$ of $\R^m$ computes a basis of the lattice $\L \cap S$.
\end{lemma}

The following theorem of Kannan~\cite{Kannan87} provides an algorithm for the Shortest Vector Problem with running time $n^{O(n)}$ and polynomial space complexity. We note that a faster algorithm with running time $2^{O(n)}$ was obtained by Micciancio and Voulgaris in~\cite{MicV10}, however its space complexity is exponential in $n$.

\begin{theorem}[\cite{Kannan87}]\label{thm:Kannan}
There exists an algorithm that given a basis of a lattice $\L$ of rank $n$ computes a shortest nonzero vector of $\L$ in running time $n^{O(n)} \cdot s^{O(1)}$ and in polynomial space, where $s$ denotes the input size.
\end{theorem}

The definition of Korkine-Zolotarev bases (Definition~\ref{def:KZ}) immediately implies that a Korkine-Zolotarev basis generating a given lattice of rank $n$ can be efficiently computed using $n$ calls to an algorithm that finds a shortest nonzero vector in a lattice. This gives us the following corollary.

\begin{corollary}\label{cor:KZalg}
There exists an algorithm that given a basis of a lattice $\L$ of rank $n$ computes a Korkine-Zolotarev basis generating $\L$ in running time $n^{O(n)} \cdot s^{O(1)}$ and in polynomial space, where $s$ denotes the input size.
\end{corollary}

Another corollary of Theorem~\ref{thm:Kannan} is the following.

\begin{corollary}\label{cor:allSValg}
There exists an algorithm that given a basis of a lattice $\L$ of rank $n$ and a number $t \geq 1$, outputs all the lattice vectors $v \in \L$ satisfying $\|v\| \leq t \cdot \lambda_1(\L)$ in running time $(t \cdot n)^{O(n)} \cdot s^{O(1)}$ and in polynomial space, where $s$ denotes the input size.
\end{corollary}

\begin{proof}
Given a lattice $\L$ of rank $n$ it is possible to compute $\lambda_1(\L)$ using Theorem~\ref{thm:Kannan} and a dual Korkine-Zolotarev basis $B$ generating $\L$ using Corollary~\ref{cor:KZalg}. Now, consider the algorithm that goes over all the linear integer combinations of the vectors in $B$ with all coefficients of absolute value at most $t \cdot n^{3/2}$ and outputs the ones that have norm at most $t \cdot \lambda_1(\L)$. The correctness of the algorithm follows from Lemma~\ref{lemma:coef}.

By Theorem~\ref{thm:Kannan} and Corollary~\ref{cor:KZalg}, the space complexity needed to compute $\lambda_1(\L)$ and $B$ is polynomial in the input size $s$, and the running time is $n^{O(n)} \cdot s^{O(1)}$. The number of iterations in the algorithm above is $(2t \cdot n^{3/2}+1)^n = (t \cdot n)^{O(n)}$. It follows that the algorithm has space complexity polynomial in $s$ and running time $(t \cdot n)^{O(n)} \cdot s^{O(1)}$, as required.
\end{proof}

\section{A Generalized Isolation Lemma}\label{sec:isolation}

In this section we prove a new generalized version of the isolation lemma of~\cite{ValiantV86}. The situation under study is the following. Let $C$ be a set of vectors in $\Z^n$ with bounded entries, and let $E:P(\Z^n) \rightarrow P(\Z^n)$ be some function from the power set of $\Z^n$ to itself, which we refer to as an {\em elimination function}. It might be useful to think of $E$ as the linear span function restricted to $\Z^n$, as for this function we will obtain Lemma~\ref{lemma:isolation_spanIntro}.

Our goal is to show that a random integer $n$-dimensional vector $z$ with bounded entries with high probability {\em uniquely} defines a sequence of vectors $x_1,\ldots,x_d$ in $C$ as follows. The vector $x_1$ is the unique vector in $C \setminus E(\emptyset)$ that achieves the minimum inner product of $z$ with vectors in $C \setminus E(\emptyset)$. Once $x_1$ is chosen, it cannot be chosen anymore, and, moreover, a certain subset of $C$, denoted $E(\{x_1\})$, is eliminated from $C$ so that its elements cannot be chosen in the next steps. Similarly, $x_2$ is the unique vector in $C \setminus E(\{x_1\})$ that achieves the minimum inner product of $z$ with vectors in $C \setminus E(\{x_1\})$, and, as before, the elements in the set $E(\{x_1,x_2\})$ cannot be chosen from now on. This process proceeds until we obtain $d$ vectors $x_1,\ldots,x_d$ which eliminate the whole $C$, that is, $C \subseteq E(\{x_1,\ldots,x_d\})$, and satisfy $x_j \in C \setminus E(\{x_1,\ldots,x_{j-1}\})$ for every $1 \leq j \leq d$.

The above process is a generalization of several known cases of the isolation lemma. For example, if the function $E$ is defined to output the empty set on itself and $\Z^n$ on every other set, then the process above will give us a single vector $x_1 \in C$ that uniquely minimizes the inner product with $z$, just like the standard isolation lemma. As another example, consider the function $E$ which is defined to act like the identity function on sets of size smaller than $d$ and to output $\Z^n$ on every other set. With this $E$ we will obtain $d$ vectors which uniquely achieve the minimum $d$ inner products of vectors in $C$ with $z$. Another example for a function $E$, which is the one used for Lemma~\ref{lemma:isolation_spanIntro}, is defined by $E(A) = \linspan(A) \cap \Z^n$. Using this elimination function we obtain $d$ linearly independent vectors $x_1,\ldots,x_d$ in $C$, such that $x_j$ uniquely achieves the minimum inner product of $z$ with vectors in $C \setminus \linspan(\{x_1,\ldots,x_{j-1}\})$ for every $1 \leq j \leq d$ where $d = \rank(\linspan(C))$.

We turn to define the type of elimination functions considered in our isolation lemma.

\begin{definition}\label{def:elimination}
For a set family $\calF \subseteq P(\Z^n)$, which is closed under intersection and satisfies $\Z^n \in \calF$, define its {\em elimination function} $E:P(\Z^n) \rightarrow P(\Z^n)$ by $E(A) = \bigcap\{X \in \calF \mid A \subseteq X\} \in \calF.$
\end{definition}

We note that all the elimination functions considered in the examples above can be defined as in Definition~\ref{def:elimination}. For the standard isolation lemma take $\calF = \{\emptyset,\Z^n\}$, for the $d$ uniquely achieved minimum inner products take $\calF = \{X \subseteq \Z^n \mid |X|<d\} \cup \{\Z^n\}$, and for the span elimination function take $\calF$ to be the family of all sets $S \cap \Z^n$ where $S$ is a linear subspace of $\R^n$. It is easy to see that all these set families are closed under intersection and include $\Z^n$.

\begin{claim}\label{claim:elimination}
Let $\calF \subseteq P(\R^n)$ be a set family as in Definition~\ref{def:elimination}, and let $E:P(\Z^n) \rightarrow P(\Z^n)$ be its elimination function. Then, for every $A,B \in P(\Z^n)$,
\begin{enumerate}
  \item\label{claim:elim} $A \subseteq E(A)$,
  \item\label{claim:contain} $A \subseteq E(B)$ implies $E(A) \subseteq E(B)$, and
  \item\label{claim:itemContain} $A \subseteq B$ implies $E(A) \subseteq E(B)$.
\end{enumerate}
\end{claim}

\begin{proof}
Item~\ref{claim:elim} is immediate from the definition of $E$. For Item~\ref{claim:contain}, assume $A \subseteq E(B)$. By the definition of $E$, $E(A)$ is contained in every set of $\calF$ which contains $A$, hence, in particular, it is contained in $E(B)$. For Item~\ref{claim:itemContain}, assume $A \subseteq B$. This implies that every set of $\calF$ which contains $B$ contains $A$ as well, therefore $E(A) \subseteq E(B)$.
\end{proof}

\begin{remark}
It can be shown that for every function $E:P(\Z^n) \rightarrow P(\Z^n)$ which satisfies Items~\ref{claim:elim} and~\ref{claim:contain} in Claim~\ref{claim:elimination} there exists a set family $\calF$ which is closed under intersection and induces $E$ as in Definition~\ref{def:elimination}.
\end{remark}

The following definition will be used in the statement of our isolation lemma.
\begin{definition}\label{def:uniquely}
For an elimination function $E:P(\Z^n) \rightarrow P(\Z^n)$ as in Definition~\ref{def:elimination} and a set $C \subseteq \Z^n$, a {\em chain} of length $d$ in $C$ is a sequence of $d$ vectors $x_1,\ldots,x_d \in C$ such that $x_j \notin E(\{x_1,\ldots,x_{j-1}\})$ for every $1 \leq j \leq d$. If, in addition, $C \subseteq E(\{x_1,\ldots,x_d\})$ we say that the chain is {\em maximal}. We say that a vector $z \in \Z^n$ {\em uniquely defines a chain} $x_1,\ldots,x_d$ in $C$ if for every $1 \leq j \leq d$, the minimum inner product of $z$ with vectors in $C \setminus E(\{x_1,\ldots,x_{j-1}\})$ is uniquely achieved by $x_j$.
\end{definition}

\begin{lemma}[A Generalized Isolation Lemma]\label{lemma:isolation}
Let $E:P(\Z^n) \rightarrow P(\Z^n)$ be an elimination function as in Definition~\ref{def:elimination}. Let $C \subseteq \Z^n$ be a set of vectors satisfying $\|c\|_\infty \leq K$ for every $c \in C$, such that every chain in $C$ has length at most $m$. Let $z=(z_1,\ldots,z_n)$ be a random vector such that each $z_i$ is independently chosen from the uniform distribution over $\{1,\ldots,R\}$ for $R = K(2K+1)m^2 n/\eps$. Then, with probability at least $1-\eps$, $z$ uniquely defines a maximal chain in $C$.
\end{lemma}

We need the following additional notations to be used in the proof.
\begin{definition}
 For a set $C \subseteq \Z^n$ and a vector $z \in \Z^n$, we let $C_z[r]$ denote the set of all vectors in $C$ whose inner product with $z$ is $r$, that is, $C_z[r] = \{ c \in C \mid \langle z,c \rangle = r \}$. For an elimination function $E:P(\Z^n) \rightarrow P(\Z^n)$ as in Definition~\ref{def:elimination}, we say that a set $C_z[r]$ is {\em contributing} to $C$ if it is not contained in the set obtained by applying $E$ to the set of vectors in $C$ whose inner product with $z$ is smaller than $r$, equivalently, $C_z[r] \nsubseteq E(\cup_{r':r'<r}{C_z[r']})$.
\end{definition}

\begin{proof}[Proof of Lemma~\ref{lemma:isolation}]
For a vector $z=(z_1,\ldots,z_n)$, we say that an index $1 \leq i \leq n$ is {\em singular} if (1) $z$ uniquely defines a chain $x_1,\ldots,x_{j}$ in $C$ for some $j \geq 0$, but (2) there are at least two vectors in $C \setminus E(\{x_1,\ldots,x_{j}\})$ that differ in the $i$th coordinate and achieve the minimum inner product with $z$ among the vectors in $C \setminus E(\{x_1,\ldots,x_{j}\})$. We prove below that for every $1 \leq i \leq n$, the probability that $i$ is singular is at most $\eps/n$. By the union bound, with probability at least $1-\eps$ none of the indices is singular, thus the lemma follows.

From now on fix an arbitrary index $1 \leq i \leq n$ and the values of $z_1,\ldots,z_{i-1},z_{i+1},\ldots,z_n$. For every $-K \leq t \leq K$ denote \[C^{(t)} = \{c \in C \mid c_i=t \}.\]
Partition every $C^{(t)}$ into the sets $C^{(t)}_z[r]$, and note that every $c,c' \in C^{(t)}$ are in the same set if and only if $\sum_{j:j \neq i}{z_j c_j}=\sum_{j:j \neq i}{z_j {c'}_j}$, independently of the value of $z_i$. Similarly, the order of the sets $C^{(t)}_z[r]$ in a non-decreasing value of $r$ is independent of $z_i$. Finally, observe that for every $t \neq t'$ and every two sets in the partitions of $C^{(t)}$ and $C^{(t')}$, there is at most one value of $z_i$ for which the inner products of $z$ with the vectors in the two sets are equal.

For every $t$ we denote by $A^{(t)}_z$ the set of all integers $r$ for which $C^{(t)}_z[r]$ is contributing to $C^{(t)}$. Using Item~\ref{claim:itemContain} of Claim~\ref{claim:elimination}, one can choose one vector from every set $C^{(t)}_z[r]$ for $r \in A^{(t)}_z$ to obtain a chain of length $|A^{(t)}_z|$. Hence, our assumption on $C$ implies that $|A^{(t)}_z| \leq m$. This gives us $2K+1$ sets $A^{(t)}_z$, each of which is of size at most $m$. Hence, there are at most $m^2 \cdot {2K+1 \choose 2} = m^2 (2K+1)K$ possible values of $z_i$ for which two distinct sets $A^{(t)}_z$ intersect. Since $z_i$ is uniformly chosen from $\{1,\ldots,R\}$, the probability that two distinct sets $A^{(t)}_z$ intersect is at most $m^2 (2K+1)K/R = \eps/n$.

To complete the proof, it suffices to show that if $i$ is singular for a vector $z$, then there exist two distinct intersecting sets $A^{(t)}_z$. Assume that $i$ is singular for a vector $z$. This implies that $z$ uniquely defines a chain $x_1,\ldots,x_{j}$ in $C$ for some $j \geq 0$, but there are two vectors $b,c$ satisfying $b_i \neq c_i$ that achieve the minimum inner product of $z$ with vectors in $C \setminus E(\{x_1,\ldots,x_{j}\})$. Partition $C$ into the sets $C_z[r]$, and let $A_z$ be the set of all integers $r$ for which $C_z[r]$ is contributing to $C$. Using Item~\ref{claim:contain} of Claim~\ref{claim:elimination}, it follows that there exists some $r \in A_z$ for which the contributing set $C_z[r]$ contains $b$ and $c$ which both do not belong to $E(\cup_{r':r'<r}{C_z[r']})$. In particular, since $C^{(t)}_z[r] \subseteq C_z[r]$ for every $t$ and $r$, Item~\ref{claim:itemContain} of Claim~\ref{claim:elimination} implies that $C^{(b_i)}_z[r]$ is contributing to $C^{(b_i)}$ and that $C^{(c_i)}_z[r]$ is contributing to $C^{(c_i)}$. Hence, $r$ belongs to both $A^{(b_i)}_z$ and $A^{(c_i)}_z$, as required.
\end{proof}

Now, we turn to derive the special case of the previous lemma, which is used in the next section (Lemma~\ref{lemma:isolation_spanIntro}). To state it, we use the following definition which is analogous to Definition~\ref{def:uniquely} for the span elimination function.

\begin{definition}\label{def:span_uniquely}
For a set $A \subseteq \R^m$ and a vector $v \in \R^m$, we say that $v$ {\em uniquely defines a linearly independent chain} of length $n$ in $A$ if there are $n$ vectors $x_1,\ldots,x_n \in A$ such that for every $1 \leq j \leq n$, the minimum inner product of $v$ with vectors in $A \setminus \linspan(x_1,\ldots,x_{j-1})$ is uniquely achieved by $x_j$.
\end{definition}

\begin{corollary}\label{cor:isolation_span}
Let $C \subseteq \Z^n$ be a set of vectors satisfying $\|c\|_\infty \leq K$ for every $c \in C$ and $\linspan(C) = \R^n$. Let $z=(z_1,\ldots,z_n)$ be a random vector such that each $z_i$ is independently chosen from the uniform distribution over $\{1,\ldots,R\}$ for $R = K(2K+1)n^3/\eps$. Then, with probability at least $1-\eps$, $z$ uniquely defines a linearly independent chain of length $n$ in $C$.
\end{corollary}

\begin{proof}
Consider the set family $\calF = \{ S \cap \Z^n \mid S \mbox{ is a subspace of }\R^n \}$. The family $\calF$ includes $\Z^n$ and is closed under intersection since subspaces of $\R^n$ are. The elimination function $E$ that $\calF$ induces is defined by $E(A) = \linspan(A) \cap \Z^n$. Observe that the vectors of every chain in $C$ (with respect to this $E$) are linearly independent, thus its length is at most $n$. Apply Lemma~\ref{lemma:isolation} with $m=n$ to obtain that the random vector $z$, with probability $1-\eps$, uniquely defines a maximal linearly independent chain in $C$. Finally, the assumption $\linspan(C) = \R^n$ implies that the length of every maximal linearly independent chain in $C$ is $n$.
\end{proof}

\section{The Algorithm}\label{sec:algorithm}

In this section we present our algorithm for $\LIP$ proving Theorem~\ref{thm:MainAlgIntro}.

\subsection{The Case \texorpdfstring{$\lambda_1=\lambda_n$}{Lambda1 = Lambdan}}

We start with the special case of lattices of rank $n$ that satisfy $\lambda_1=\lambda_n$ (i.e., contain $n$ linearly independent shortest vectors), and prove the following.

\begin{theorem}\label{thm:SpecialAlg}
There exists an algorithm that given two bases of lattices $\L_1$ and $\L_2$ of rank $n$ satisfying $\lambda_1=\lambda_n$, outputs all orthogonal linear transformations $O:\linspan(\L_1) \rightarrow \linspan(\L_2)$ for which $\L_2 = O(\L_1)$ in running time $n^{O(n)} \cdot s^{O(1)}$ and in polynomial space, where $s$ denotes the input size. In addition, the number of these transformations is at most $n^{O(n)}$.
\end{theorem}

The algorithm that implies Theorem~\ref{thm:SpecialAlg} relies on the following theorem (recall Definition~\ref{def:span_uniquely}).

\begin{theorem}\label{thm:uniquelyL}
Let $\L$ be a lattice of rank $n$ satisfying $\lambda_1(\L)=\lambda_n(\L)$, and let $A$ denote the set of all shortest nonzero vectors of $\L$. Then there exists a vector $v \in \L^*$ that uniquely defines a linearly independent chain of length $n$ in $A$ and satisfies $\|v\| \leq 5n^{17/2} \cdot \lambda_1(\L^*)$.
\end{theorem}

\begin{proof}
Let $B$ be a dual Korkine-Zolotarev basis generating the lattice $\L$, and let $C$ be the set of coefficients of shortest nonzero vectors of $\L$ in terms of the basis $B$, that is, \[C = \{ x \in \Z^n \mid Bx \in A\}.\]
Observe that Lemma~\ref{lemma:coef} applied with $t=1$ implies that all the entries of the integer vectors in $C$ have absolute value at most $n^{3/2}$. Since $\lambda_1(\L)=\lambda_n(\L)$, $A$ contains $n$ linearly independent vectors, hence their coefficient vectors in $C$ are linearly independent as well, so $\linspan(C)=\R^n$. We apply the isolation lemma (Corollary~\ref{cor:isolation_span}) with $K=n^{3/2}$ and, say, $\eps=1/2$. We obtain that for $R = 2K(2K+1)n^3 \leq 5 n^6$, there exists a vector $z \in \{1,\ldots,R\}^n$ that uniquely defines a linearly independent chain of length $n$ in $C$. Since $\langle x,y \rangle  = \langle Bx,B^*y \rangle$ for every $x,y \in \R^n$, it follows that the vector \[v = B^*z = \sum_{i=1}^{n}{z_i \cdot b^*_i} \in \L^*\] uniquely defines a linearly independent chain of length $n$ in $A$. Finally, since $B$ is a dual Korkine-Zolotarev basis generating $\L$ and $\lambda_1(\L) = \lambda_n(\L)$, Lemma~\ref{lemma:normK} implies that \[\|v\| \leq n^{5/2} \cdot 5n^{6} \cdot \lambda_1(\L^*) = 5n^{17/2} \cdot \lambda_1(\L^*). \qedhere \]
\end{proof}

\begin{proof}[Proof of Theorem~\ref{thm:SpecialAlg}]
Let $\L_1$ and $\L_2$ be the lattices generated by the input bases $B_1$ and $B_2$. Consider the algorithm that acts as follows (see Algorithm~\ref{alg:ISO_new}). For $i \in \{1,2\}$, the algorithm computes the set $A_i$ of all shortest nonzero vectors of $\L_i$ and the set $W_i$ of all vectors in the dual lattice $\L_i^*$ of norm at most $5n^{17/2} \cdot \lambda_1(\L_i^*)$. These sets can be computed using the algorithm from Corollary~\ref{cor:allSValg}. Given these sets, the algorithm finds a $w_1 \in W_1$ that uniquely defines a linearly independent chain of length $n$ in $A_1$ and the corresponding chain $x_1,\ldots,x_n \in A_1$. The existence of $w_1$ is guaranteed by Theorem~\ref{thm:uniquelyL}. Then, the algorithm goes over all vectors $w_2 \in W_2$ and for every $w_2$ which uniquely defines a linearly independent chain $y_1,\ldots,y_n \in A_2$ it checks if the linear transformation $O:\linspan(\L_1) \rightarrow \linspan(\L_2)$, defined by $O(x_i)=y_i$ for every $1 \leq i \leq n$, is orthogonal and maps $\L_1$ to $\L_2$. If this is the case, then $O$ is inserted to the output set.

\begin{algorithm}[ht]
    \caption{Lattice Isomorphism -- Special Case}
    \textbf{Input:} Two bases of lattices $\L_1$ and $\L_2$ of rank $n$ satisfying $\lambda_1(\L_1) = \lambda_n(\L_1)$ and $\lambda_1(\L_2) = \lambda_n(\L_2)$. \\
    \textbf{Output:} The set $\Output$ of all orthogonal linear transformations $O:\linspan(\L_1) \rightarrow \linspan(\L_2)$ for which $\L_2 = O(\L_1)$.
    \begin{algorithmic}[1]
        \ForAll{$i=1,2$}
            \State{$A_i \leftarrow \{ x \in \L_i \mid \|x\| = \lambda_1(\L_i) \}$}\label{line:A_i}\Comment{Corollary~\ref{cor:allSValg}}
            \State{$W_i \leftarrow \{ x \in \L^*_i \mid \|x\| \leq 5n^{17/2} \cdot \lambda_1(\L^*_i) \}$}\label{line:W_i}\Comment{Corollary~\ref{cor:allSValg}}
        \EndFor
        \ForAll{$w_1 \in W_1$}
            \If{$w_1$ uniquely defines a linearly independent chain of length $n$ in $A_1$}
                \State{$(x_1,\ldots,x_n) \leftarrow $ the maximal linearly independent chain that $w_1$ uniquely defines in $A_1$}
                \State{\textbf{goto} line~\ref{line:second_vec}}
            \EndIf
        \EndFor
        \ForAll{$w_2 \in W_2$}\label{line:second_vec}
            \If{$w_2$ uniquely defines a linearly independent chain of length $n$ in $A_2$}
                \State{$(y_1,\ldots,y_n) \leftarrow $ the maximal linearly independent chain that $w_2$ uniquely defines in $A_2$}
                \State{$O \leftarrow$ the linear transformation that maps $x_i$ to $y_i$ for every $1 \leq i \leq n$}\label{line:Odef}
                \If{$O$ is orthogonal and satisfies $\L_2=O(\L_1)$}
                    \State{$\Output \leftarrow \Output \cup \{O\}$}\label{line:O}
                \EndIf
            \EndIf
        \EndFor
    \end{algorithmic}
    \label{alg:ISO_new}
\end{algorithm}

We turn to prove the correctness of the algorithm. It is clear from the algorithm that any linear transformation in the output is orthogonal and maps $\L_1$ to $\L_2$. We claim that every orthogonal linear transformation that maps $\L_1$ to $\L_2$ is in the output. To see this, let $O:\linspan(\L_1) \rightarrow \linspan(\L_2)$ be such a transformation. Consider the vector $u = O(w_1)$ where $w_1 \in \L_1^*$ is the vector which is computed by the algorithm and uniquely defines a linearly independent chain $x_1,\ldots,x_n$ in $A_1$. Since $O$ preserves inner products, it follows that $u \in \L_2^*$ and that
\[\|u\| = \|w_1\| \leq 5n^{17/2} \cdot \lambda_1(\L_1^*) = 5n^{17/2} \cdot \lambda_1(\L_2^*).\]
Therefore, $u$ belongs to $W_2$. Since $A_2 = O(A_1)$, it follows that $u$ uniquely defines a linearly independent chain of length $n$ in $A_2$, and that this chain is $O(x_1),\ldots,O(x_n)$. Thus, the chain $y_1,\ldots,y_n$, which is computed by the algorithm for $u$, satisfies $O(x_i)=y_i$ for every $1 \leq i \leq n$, so the algorithm includes $O$ in its output.

Now we analyze the running time and the space complexity of Algorithm~\ref{alg:ISO_new}. We start with the running time, and focus on its dependence on the rank $n$, ignoring terms which are polynomial in the input size $s$. By Corollary~\ref{cor:allSValg}, the running time needed to compute the sets $A_i$ is $n^{O(n)}$ and to compute the sets $W_i$ is $(5n^{17/2} \cdot n)^{O(n)} = n^{O(n)}$. By Fact~\ref{fact:short}, we have $|A_i| = 2^{O(n)}$ and $|W_i| = n^{O(n)}$. Given a vector $w$ and a set $A$, it is possible to check in time polynomial in $|A|$ and in the input size if $w$ uniquely defines a linearly independent chain of length $n$ in $A$, and if so to compute the chain. Hence the total running time is $n^{O(n)}$. For the space complexity of the algorithm recall that the algorithm from Corollary~\ref{cor:allSValg} requires only polynomial space. In order to have only polynomial space complexity in Algorithm~\ref{alg:ISO_new}, it should be implemented in a way that the sets $A_i$ and $W_i$ are not stored at any step of the algorithm. Instead, whenever the algorithm checks if a vector uniquely defines a linearly independent chain of length $n$ in a set $A_i$, this set should be recomputed. Since the number of calls to this procedure is $n^{O(n)}$, the running time remains $n^{O(n)}$. Similarly, the sets $W_i$ should not be stored, as it suffices to enumerate their elements in order to implement the algorithm. Therefore, the algorithm can be implemented in a way that requires only polynomial space complexity and the stated running time. Finally, observe that the number of returned orthogonal linear transformations is bounded from above by $|W_2|$, hence is at most $n^{O(n)}$.
\end{proof}

\subsection{The General Case}

Now, we turn to deal with the general case, where the successive minima of the input lattices are not necessarily all equal. We start with the following simple lemma.

\begin{lemma}\label{lemma:general}
Let $\L_1$ and $\L_2$ be two lattices, and let $O:\linspan(\L_1) \rightarrow \linspan(\L_2)$ be an orthogonal linear transformation satisfying $\L_2 = O(\L_1)$. For $i \in\{1,2\}$, let $V_i$ be the linear subspace spanned by all shortest nonzero vectors of $\L_i$, and let $\pi_i$ denote the projection of $\linspan(\L_i)$ to the orthogonal complement to $V_i$. Then,
\begin{enumerate}
  \item The restriction $O|_{V_1}$ of $O$ to $V_1$ is an orthogonal linear transformation mapping the lattice $\L_1 \cap V_1$ to the lattice $\L_2 \cap V_2$. In addition, the lattices $\L_1 \cap V_1$ and $\L_2 \cap V_2$ have the same rank $k$ and they both satisfy $\lambda_1=\lambda_k$.
  \item The restriction $O|_{\pi_1(\linspan(\L_1))}$ of $O$ to $\pi_1(\linspan(\L_1))$ is an orthogonal linear transformation mapping the lattice $\pi_1(\L_1)$ to the lattice $\pi_2(\L_2)$.
\end{enumerate}
\end{lemma}

\begin{proof}
Since $O$ preserves lengths, $v$ is a shortest nonzero vector of $\L_1$ if and only if $O(v)$ is a shortest nonzero vector of $\L_2$, thus $O(V_1)=V_2$. This implies that $O$ satisfies $O(\L_1 \cap V_1) = \L_2 \cap V_2$, so does its restriction $O|_{V_1}$. Therefore, the lattices $\L_1 \cap V_1$ and $\L_2 \cap V_2$ are isomorphic and, in particular, have the same rank $k$. Since these lattices contain $k$ linearly independent shortest nonzero vectors, it follows that they both satisfy $\lambda_1=\lambda_k$. Now, observe that every $x \in \linspan(\L_1)$ satisfies $O(\pi_1(x)) = \pi_2(O(x))$. Hence, \[O(\pi_1(\L_1)) = \pi_2(O(\L_1)) = \pi_2(\L_2),\] so $O|_{\pi_1(\linspan(\L_1))}$ is an orthogonal linear transformation mapping $\pi_1(\L_1)$ to $\pi_2(\L_2)$.
\end{proof}

Equipped with Lemma~\ref{lemma:general}, Theorem~\ref{thm:MainAlgIntro} follows quite easily.

\begin{proof}[Proof of Theorem~\ref{thm:MainAlgIntro}]

Let $\L_1$ and $\L_2$ be the lattices generated by the input bases $B_1$ and $B_2$. Consider the algorithm that acts as follows (see Algorithm~\ref{alg:IsoGeneral}). For $i \in \{1,2\}$, the algorithm computes the linear subspace $V_i$ spanned by all shortest nonzero vectors of $\L_i$ and the projection $\pi_i$ of $\linspan(\L_i)$ to the orthogonal complement to $V_i$. This can be done using the algorithm from Corollary~\ref{cor:allSValg} for computing shortest nonzero vectors of a given lattice. If the lattices $\L_1 \cap V_1$ and $\L_2 \cap V_2$ do not have the same rank, then the algorithm outputs that the lattices $\L_1$ and $\L_2$ are not isomorphic. Otherwise, the algorithm computes, using the algorithm from Theorem~\ref{thm:SpecialAlg}, all orthogonal linear transformations $O_1$ that map $\L_1 \cap V_1$ to $\L_2 \cap V_2$, and recursively computes all orthogonal linear transformations $O_2$ that map $\pi_1(\L_1)$ to $\pi_2(\L_2)$. Finally, the algorithm checks for every such pair $(O_1,O_2)$ if the transformation $O$, defined on $\linspan(\L_1)$ by $O|_{V_1} = O_1$ and $O|_{\pi_1(\linspan(\L_1))}= O_2$, maps $\L_1$ to $\L_2$, and if so, inserts it to the output set.

\begin{algorithm}[ht]
    \caption{Lattice Isomorphism -- General Case}
    \textbf{Input:} Two bases of lattices $\L_1$ and $\L_2$ of rank $n$. \\
    \textbf{Output:} The set $\Output$ of all orthogonal linear transformations $O: \linspan(\L_1) \rightarrow \linspan(\L_2)$ for which $\L_2 = O(\L_1)$.
    \begin{algorithmic}[1]
        \ForAll{$i=1,2$}\label{line:loop_i}
            \State{$V_i \leftarrow \linspan(\{ x \in \L_i \mid \|x\| = \lambda_1(\L_i)\})$}\label{line:V_i}\Comment{Corollary~\ref{cor:allSValg}}
            \State{$\pi_i \leftarrow $ the projection of $\linspan(\L_i)$ to the orthogonal complement to $V_i$}\label{line:pi_i}
        \EndFor
        \If{$\rank(\L_1 \cap V_1) \neq \rank(\L_2 \cap V_2)$}
            \State{\textbf{return} $\emptyset$}\label{line:empty}
        \EndIf
        \State{$\Output_1 \leftarrow $ Lattice Isomorphism Special Case$(\L_1 \cap V_1,\L_2 \cap V_2)$}\label{line:callSpecial}\Comment{Algorithm~\ref{alg:ISO_new} (Theorem~\ref{thm:SpecialAlg})}
        \State{$\Output_2 \leftarrow $ Lattice Isomorphism General Case$(\pi_1(\L_1),\pi_2(\L_2))$}\label{line:callGeneral}\Comment{A recursive call}
        \State{$\Output \leftarrow \emptyset$}
        \ForAll{$O_1 \in \Output_1,O_2 \in \Output_2$}\label{line:loop12}
            \State{$O \leftarrow $ the linear transformation defined on $\linspan(\L_1)$ by $O|_{V_1} = O_1$ and $O|_{\pi_1(\linspan(\L_1))} = O_2$}\label{line:def_O}
            \If{$O$ satisfies $\L_2 = O(\L_1)$}\label{line:if_O_g}
                \State{$\Output \leftarrow \Output \cup \{O\}$}\label{line:O_g}
            \EndIf
        \EndFor
        \State{\textbf{return} $\Output$}
    \end{algorithmic}
    \label{alg:IsoGeneral}
\end{algorithm}

It is easy to see that the rank of the input lattices decreases in every recursive call of the algorithm. Therefore, the algorithm terminates, and its correctness follows from Lemma~\ref{lemma:general}.

We turn to show that the running time of Algorithm~\ref{alg:IsoGeneral} on lattices of rank $n$ is $n^{O(n)} \cdot s^{O(1)}$, where $s$ denotes the input size. As before, we ignore in the analysis terms which are polynomial in $s$. Denote by $r \leq n$ the number of recursive calls, let $n_i$ denote the rank of the input lattices of the $i$th recursive call, and observe that $\sum_{i=1}^{r}{n_i} = n$. We analyze the total running time of every step of Algorithm~\ref{alg:IsoGeneral} in all the $r$ recursive calls together.

Using Corollary~\ref{cor:allSValg}, it can be shown that the running time of computing the subspaces $V_1$ and $V_2$ and the projections $\pi_1$ and $\pi_2$ in the $i$th recursive call is $n_i^{O(n_i)}$. Hence, the total running time of the loop in line~\ref{line:loop_i} is $n^{O(n)}$. Given $V_i$ and $\pi_i$, by Lemmas~\ref{lemma:MicSubspace} and~\ref{lemma:MinimalLattice}, it is possible to compute in polynomial time bases for the lattices $\L_i \cap V_i$ and $\pi_i(\L_i)$. By Theorem~\ref{thm:SpecialAlg}, the output of Algorithm~\ref{alg:ISO_new} (line~\ref{line:callSpecial}) in the $i$th recursive call has size $n_i^{O(n_i)}$, and its computation requires running time $n_i^{O(n_i)}$. This implies that the total running time of the calls to Algorithm~\ref{alg:ISO_new} is $n^{O(n)}$, and that the total running time of the loop in line~\ref{line:loop12} is at most
\[\prod_{i=1}^{r}{n_i^{O(n_i)}} \leq \prod_{i=1}^{r}{n^{O(n_i)}} = n^{O(n)},\]
so the total running time of Algorithm~\ref{alg:IsoGeneral} is bounded by $n^{O(n)} \cdot s^{O(1)}$, as required. In addition, it follows that the number of linear transformations that the algorithm outputs is at most $n^{O(n)}$.

We finally note that it is not difficult to see that Algorithm~\ref{alg:IsoGeneral} can be implemented in polynomial space and in running time as before. To do so, for computing $V_i$ (line~\ref{line:V_i}) one has to enumerate the shortest nonzero vectors of $\L_i$ and to store only the ones which are linearly independent of the previously stored ones. Similarly, in the loop of line~\ref{line:loop12}, the elements of the $\Output_i$'s should not be stored but should be recursively enumerated in parallel. Since the depth of the recursion is at most $n$, all the linear transformations which together define a purported $O$ can be simultaneously stored in space complexity polynomial in the input size.
\end{proof}

\section{The Lattice Isomorphism Problem is in \texorpdfstring{$\SZK$}{SZK}}\label{sec:coAM}

In this section we present an $\SZK$ proof system for the complement of $\LIP$ implying Theorem~\ref{thm:coAMIntro}. To do so, we need some properties of the discrete Gaussian distribution on lattices, proven in the following section.

\subsection{Gaussian-Distributed Generating Sets}

In Lemma~\ref{lemma:SampleGenSet} below we bound the number of samples from the discrete Gaussian distribution $D_{\L,s}$ needed in order to get a generating set of $\L$ with high probability. We start with the following lemma.

\begin{lemma}\label{lemma:sublattice}
For every lattice $\L$ of rank $n$ and a strict sublattice $\calM \subsetneq \L$,
\begin{enumerate}
  \item\label{item:sublattice_lambdan} If $\linspan(\calM) \subsetneq \linspan(\L)$ and $s \geq c \cdot \lambda_n(\L)$ then $\Prob{x \in D_{\L,s}}{x \in \calM} \leq \frac{1}{1+e^{- \pi c^{-2}}}$.
  \item\label{item:sublattice_bl} If $s \geq c \cdot bl(\L)$ then $\Prob{x \in D_{\L,s}}{x \in \calM} \leq \frac{1}{1+e^{- \pi c^{-2}}}$.
\end{enumerate}
\end{lemma}

\begin{proof}
For a lattice $\L$ and a strict sublattice $\calM$, let $w$ be a vector in $\L \setminus \calM$. Then, the lattice $\calM$ and its coset $w+\calM$ are disjoint and are both contained in $\L$. Using Claim~\ref{claim:GaussianBound}, we obtain that
\begin{eqnarray*}
\rho_s(\L) \geq \rho_s(\calM) + \rho_s(w+\calM) \geq (1+\rho_s(w)) \cdot \rho_s(\calM),
\end{eqnarray*}
which implies that
\[ \Prob{x \in D_{\L,s}}{x \in \calM} = \frac{\rho_s(\calM)}{\rho_s(\L)} \leq \frac{1}{1+\rho_s(w)}.\]

For Item~\ref{item:sublattice_lambdan}, observe that if $\linspan(\calM) \subsetneq \linspan(\L)$ then there exists a vector $w \in \L \setminus \calM$ such that $\|w\| \leq \lambda_n(\L)$. Applying the above argument with this $w$ for $s \geq c \cdot \lambda_n(\L)$ completes the proof. For Item~\ref{item:sublattice_bl}, recall that the lattice $\L$ is generated by a basis all of whose vectors are of norm at most $bl(\L)$. Since $\calM$ is a strict sublattice of $\L$ at least one of these vectors does not belong to $\calM$, so there exists a vector $w \in \L \setminus \calM$ such that $\|w\| \leq bl(\L)$. Apply again the above argument for $s \geq c \cdot bl(\L)$, and we are done.
\end{proof}

\begin{remark}
Note that the bounds given in Lemma~\ref{lemma:sublattice} converge to $1/2$ as the parameter $s$ increases.
\end{remark}

The following corollary resembles Corollary~3.16 in~\cite{Regev05}.
\begin{corollary}\label{cor:Ind}
For every lattice $\L$ of rank $n$ and $s \geq \lambda_n(\L)$, the probability that a set of $n^2$ vectors chosen independently from $D_{\L,s}$ contains no $n$ linearly independent vectors is $2^{-\Omega(n)}$.
\end{corollary}

\begin{proof}
Let $u_1,\ldots,u_{n^2}$ denote $n^2$ samples from $D_{\L,s}$. For every $1 \leq i \leq n$ let $A_i$ be the event that
\[\rank(\linspan(u_1,\ldots,u_{(i-1)n})) = \rank(\linspan(u_1,\ldots,u_{in})) < n.\]
Fix some $i$ and condition on a fixed choice of $u_1,\ldots,u_{(i-1)n}$ that span a subspace of rank smaller than $n$. Observe that, by Item~\ref{item:sublattice_lambdan} of Lemma~\ref{lemma:sublattice}, the probability that
\[u_{(i-1)n+1},\ldots, u_{in} \in \L \cap \linspan(u_1,\ldots,u_{(i-1)n})\]
is at most $(1+e^{-\pi})^{-n} = 2^{-\Omega(n)}$. This implies that the event $A_i$ happens with probability $2^{-\Omega(n)}$. Therefore, except with probability $2^{-\Omega(n)}$, none of the $A_i$'s happens, thus $u_1,\ldots,u_{n^2}$ contain $n$ linearly independent vectors.
\end{proof}

\begin{lemma}\label{lemma:SampleGenSet}
For every lattice $\L$ of rank $n$ satisfying $\det(\L) \geq 1$ and every $s \geq bl(\L)$, the probability that a set of \[n^2+n\log(s\sqrt{n})(n+20\log\log (s\sqrt{n}))\] vectors chosen independently from $D_{\L,s}$ does not generate $\L$ is $2^{-\Omega(n)}$.
\end{lemma}

\begin{proof}
Let $u_1,\ldots,u_{n^2}$ denote the first $n^2$ samples from $D_{\L,s}$. Since $s \geq bl(\L) \geq \lambda_n(\L)$, Corollary~\ref{cor:Ind} implies that, except with probability $2^{-\Omega(n)}$, they contain $n$ linearly independent vectors. By Lemma~\ref{lemma:BanaNorm} and the union bound, with a similar probability, each of these vectors has norm at most $s \cdot \sqrt{n}$. Denote by $\calM_0 \subseteq \L$ the sublattice generated by $u_1,\ldots,u_{n^2}$. By the assumption $\det(\L) \geq 1$, we obtain that the index of $\calM_0$ in $\L$ satisfies \[|\L : \calM_0| = \frac{\det(\calM_0)}{\det(\L)} \leq (s \cdot \sqrt{n})^n.\]

Now, define $h = n\log(s\sqrt{n})$ and $\ell = n+20\log \log (s\sqrt{n})$, and let $v_1,\ldots,v_{h \cdot \ell}$ be the remaining $h \cdot \ell$ vectors chosen from $D_{\L,s}$. For every $1 \leq i \leq h$ define the sublattice \[\calM_i = \L(u_1,\ldots,u_{n^2},v_1,\ldots,v_{i \cdot \ell}) \subseteq \L,\]
and let $A_i$ be the event that $\calM_{i-1} = \calM_{i} \subsetneq \L$. If none of the $A_i$'s happens, then $M_i = \L$ for some $i$ or $|\calM_{i} : \calM_{i-1}| \geq 2$ for every $i$. In the latter case it follows that \[|\L : \calM_{h}| \leq (s\sqrt{n})^n \cdot 2^{-h} = 1.\]
Therefore, in both cases the lattice $\calM_h$, which is generated by the $n^2+h \cdot \ell$ samples from $D_{\L,s}$, equals the lattice $\L$.

In order to complete the proof it remains to show that the probability of every event $A_i$ is at most $2^{-\Omega(n)-\log \log (s\sqrt{n})}$, as this implies by the union bound that the probability that at least one $A_i$ happens is at most \[h \cdot 2^{-\Omega(n)-\log \log (s\sqrt{n})} \leq 2^{-\Omega(n)}.\] So fix some $i$, condition on a fixed choice of $v_1,\ldots,v_{(i-1)\ell}$ which do not generate $\L$, and apply Item~\ref{item:sublattice_bl} of Lemma~\ref{lemma:sublattice} to obtain that the probability that all the vectors $v_{(i-1) \cdot \ell +1},\ldots,v_{i \cdot \ell}$ belong to $\calM_{i-1}$ is at most $(1+e^{-\pi})^{-\ell} \leq 2^{-\Omega(n)-\log \log (s\sqrt{n})}$. Thus, this is an upper bound on the probability that the event $A_i$ happens, so we are done.
\end{proof}

\subsection{\texorpdfstring{$\LIP$}{LIP} is in \texorpdfstring{$\SZK$}{SZK}}

\begin{theorem}\label{thm:LIPcoAM2}
$\LIP$ is in $\SZK$.
\end{theorem}

\begin{proof}
It is sufficient to prove that the {\em complement} of $\LIP$ has a statistical zero-knowledge proof system with respect to a honest verifier ($\HVSZK$)~\cite{Okamoto00}. This follows from the $\HVSZK$ proof system given in Algorithm~\ref{alg:coAM}. Let $B_1$ and $B_2$ be two bases of lattices of rank $n$, and define \[s = \max(\|B_1\|,\|B_2\|) \cdot \sqrt{\ln(2n+4)/\pi},\]
where $\|B\|$ denotes the norm of a longest vector in a basis $B$. It can be assumed without loss of generality that the lattices have determinant $1$, since if their determinants are distinct then they are clearly not isomorphic, and otherwise they can be scaled to have determinant $1$.

In the proof system, the verifier chooses uniformly at random an $i \in \{1,2\}$ and sends to the prover the Gram matrix $G = W^T \cdot W$, where the columns of $W$ are \[N = n^2+n\log(s\sqrt{n})(n+20\log\log (s\sqrt{n}))\] lattice vectors of $\L(B_i)$ independently chosen from $D_{\L(B_i),s}$ using the algorithm $\mathsf{SampleD}$ from Lemma~\ref{lemma:SampleD}. Finally, the verifier accepts if and only if the prover correctly guesses $i$.

\begin{algorithm}[ht]
    \caption{An $\HVSZK$ Proof System for the complement of $\LIP$}
    \textbf{Input:} Two bases $B_1$ and $B_2$ of lattices of rank $n$ with determinant $1$.
    \begin{algorithmic}[1]
        \State{$s \leftarrow \max(\|B_1\|,\|B_2\|) \cdot \sqrt{\ln(2n+4)/\pi}$}\Comment{Lemma~\ref{lemma:SampleD}}
        \State{Verifier chooses uniformly at random $i \in \{1,2\}$}
        \State{$N \leftarrow n^2+n\log(s\sqrt{n})(n+20\log\log (s\sqrt{n}))$}\Comment{Lemma~\ref{lemma:SampleGenSet}}
        \ForAll{$1 \leq j \leq N$}
            \State{$w_j \leftarrow \mathsf{SampleD}(B_i,s)$}\Comment{Lemma~\ref{lemma:SampleD}}
        \EndFor
        \State{Verifier sends $G = W^T \cdot W$ to the prover}
        \State{Prover returns $i' \in \{1,2\}$}
        \State{Verifier accepts if and only if $i=i'$}
    \end{algorithmic}
    \label{alg:coAM}
\end{algorithm}

By Lemma~\ref{lemma:SampleGenSet}, except with probability $2^{-\Omega(n)}$, the set $W$ generates the lattice $\L(B_i)$, so for simplicity we assume from now on that this is the case. The running time needed by the verifier in the protocol is clearly polynomial in the input size. So we turn to prove correctness, that is, that for every prover's strategy the verifier rejects $\YES$ instances with some non-negligible probability, whereas $\NO$ instances are accepted for some prover's strategy.

Assume that $(B_1,B_2)$ is a $\YES$ instance. This means that the lattices $\L(B_1)$ and $\L(B_2)$ are isomorphic, so there exists an orthogonal linear transformation $O:\linspan(B_1) \rightarrow \linspan(B_2)$ mapping $\L(B_1)$ to $\L(B_2)$. Recall that the matrix $W$ is chosen either from $D_{\L(B_1),s}^N$ or from $D_{O(\L(B_1)),s}^N$. Since $O$ preserves inner products, the Gram matrix of $W$ equals the Gram matrix of $O(W)$. Therefore, the distribution of the matrix $G$, which is sent to the prover, is independent of $i$. Hence, the probability that the verifier accepts is at most $1/2$.

Now, assume that $(B_1,B_2)$ is a $\NO$ instance, and consider the following strategy for the computationally unbounded prover: Given $G$, the prover returns the $i'$ for which there exists a vector set $W$ of size $N$ that generates $\L(B_{i'})$ and satisfies $G = W^T \cdot W$. To complete the proof it remains to show that this $i'$ is unique whenever $\L(B_1)$ and $\L(B_2)$ are not isomorphic. Indeed, if $W_1^T \cdot W_1 = W_2^T \cdot W_2$ then by Fact~\ref{fact:Gram} there exists an orthogonal linear transformation $O$ mapping $W_1$ to $W_2$, and this implies that the lattices generated by the sets $W_1$ and $W_2$ are isomorphic.

To complete the proof, it remains to observe that the presented proof system is statistical zero-knowledge, that is, the honest verifier ``learns nothing'' from the interaction with the prover other than the fact that the lattices are not isomorphic. To see this, consider the probabilistic polynomial-time simulator that runs the proof system playing the roles of both the honest verifier and the prover, where as prover it returns $i'$ which equals the $i$ chosen by the simulated verifier. Observe that the distribution of the transcript obtained from this simulation is statistically close to the one obtained from a run by a honest verifier and a prover.
\end{proof}

\begin{remark}
We remark that by fixing the answer of the prover in the unlikely event that the set $W$ does not generate the lattice, it follows that the complement of $\LIP$ has a honest verifier {\em perfect} zero-knowledge proof system.
\end{remark}

\section*{Acknowledgement}
We would like to thank Frank Vallentin for useful comments.

\bibliographystyle{abbrv}
\bibliography{latticeiso}


\end{document}